\documentclass{article}

\usepackage{arxiv}

\usepackage[utf8]{inputenc} 
\usepackage[T1]{fontenc}    
\usepackage{hyperref}       
\usepackage{url}            
\usepackage{booktabs}       
\usepackage{amsfonts}       
\usepackage{nicefrac}       
\usepackage{microtype}      
\usepackage{lipsum}
\usepackage{graphicx}
\graphicspath{ {./images/} }

\usepackage{cite}
\usepackage{amsmath,amssymb,amsfonts}
\usepackage{algorithmic}
\usepackage{graphicx}
\usepackage{textcomp}
\usepackage{booktabs}
\usepackage{amsmath}
\usepackage{pdfpages}
\usepackage{hyperref}
\usepackage{amsthm}
\usepackage[ruled,vlined,linesnumbered,resetcount]{algorithm2e}
\theoremstyle{plain}
\newtheorem{thm}{Theorem} 

\theoremstyle{definition}
\newtheorem{defn}{Definition}  
\newtheorem{exmp}{Example}  
\newtheorem{problem}{Problem}
\newtheorem{prop}{Proposition}

\title{PrivGen: Preserving Privacy of Sequences Through Data Generation}

\author{
 Sigal Shaked \\
 Dept. of Software and Information Systems Eng.\\ Ben-Gurion University of the Negev\\ Be'er Sheva, Israel\\
  \texttt{shaksi@post.bgu.ac.il} \\
   \And
   Lior Rokach \\
 Dept. of Software and Information Systems Eng.\\ Ben-Gurion University of the Negev\\ Be'er Sheva, Israel\\
 \texttt{liorrk@bgu.ac.il} \\
}

\begin{document}
\maketitle
\begin{abstract}
Sequential data is everywhere, and it can serve as a basis for research that will lead to improved processes. For example, road infrastructure can be improved by identifying bottlenecks in GPS data, or early diagnosis can be improved by analyzing patterns of disease progression in medical data. The main obstacle is that access and use of such data is usually limited or not permitted at all due to concerns about violating user privacy, and rightly so. Anonymizing sequence data is not a simple task, since a user creates an almost unique signature over time. Existing anonymization methods reduce the quality of information in order to maintain the level of anonymity required. Damage to quality may disrupt patterns that appear in the original data and impair the preservation of various characteristics. Since in many cases the researcher does not need the data as is and instead is only interested in the patterns that exist in the data, we propose PrivGen, an innovative method for generating data that maintains patterns and characteristics of the source data. We demonstrate that the data generation mechanism significantly limits the risk of privacy infringement. Evaluating our method with real-world datasets shows that its generated data preserves many characteristics of the data, including the sequential model, as trained based on the source data. This suggests that the data generated by our method could be used in place of actual data for various types of analysis, maintaining user privacy and the data’s integrity at the same time. 
\end{abstract}


\section{Introduction}
\label{sec:introduction}
Data sharing is essential for leveraging the opportunities associated with the flood of big data. Publishing sequential datasets results in vast amounts of data that can be utilized for learning about the behavior of individuals, populations, and trends. In some cases, the exact data is required, and in other cases, only the characteristics of the data are important \cite{wang2010privacy}. For example, a researcher who examines the effect of a drug on patients will need the actual data in order to identify other side effects experienced by relevant patients. In contrast, a telecom company that collects user device locations but cannot use them due to existing regulations can use other data with similar characteristics to help improve various processes within the company, including the positioning of antennas and customization of services to subscribers. With the emergence of various applications, such as navigation systems, electronic medical record (EMR) systems \cite{gkoulalas2014publishing}, and Web search engines, the amount of sequential data collected is enormous. Despite the potential benefits, the current climate of threats and attackers, and limited privacy protection, does not encourage the publication and sharing of data.

In addition, there is pervasive fear and concern about potential attackers that are becoming increasingly sophisticated and gaining access to user data. Real-world examples include recent attempts to cross trajectories collected by mobile operators with credit card records \cite{riederer2016linking}, geographic check-ins of Twitter and Flickr users \cite{cecaj2014re}, and Google Places, Yelp, and Facebook metadata \cite{mayer2016evaluating}. Such attacks make data owners fearful of taking the risk and cause them to refrain from publishing their data.

In this atmosphere, data protection legislation in Europe was updated, by replacing the 1995 data protection provision with the General Data Protection Regulation (GDPR). The GDPR responds to threats of privacy and contributes to data protection by strengthening individual protections and increases the need for robust solutions that enable data sharing while ensuring individuals' privacy.

Various anonymization methods have been suggested to address privacy threats. K-anonymity based techniques make basic assumptions about the prior knowledge the attacker possesses and thereby provide just limited privacy protection. Protection is even more limited in the case of sequential data, where the sequence itself, which is often multi-dimensional, is a quasi-identifier. For example, an attacker that knows the locations visited by a certain user may identify the user by searching for users that visited those locations. Differential privacy-based techniques provide wider protection as no assumptions are made regarding the attacker’s advance knowledge, but their practical application may have a negative impact on data utility, due to the loss of information caused by the data transformation. 

Synthetic data can be thought of as "fake" data created from "real" data. The beauty of it comes from its foundation in real data and real distributions, which makes it almost indistinguishable from the original data. The "usefulness" of synthetic data has been validated by studies like  \cite{bellovin2019privacy,patki2016synthetic,choi2017generating,dutta2018simulated}. Its momentum, in the context of privacy, stems from the fact that many times it is legally impossible to share real data, but in fact, anonymized data is insufficiently useful. At these moments, establishing a synthetic dataset may present the best solution of both worlds - data that can be shared and yet resembles the original.

We propose a synthetic data generation approach for data publication: PrivGen, a data generation-based approach which is appropriate for scenarios where real details are not needed, and data that behaves like the original data is acceptable. When publishing synthesized datasets, the data publisher must ensure that no sensitive information about specific individuals is disclosed. Synthetic datasets intuitively seem to mitigate the re-identification problem, since published records are invented and do not derive from any original record \cite{hundepool2012statistical}. Whereas there is a concern about membership inference (\cite{patki2016synthetic, yeom2018privacy, long2018understanding} ), i.e. knowing that a single item exists in the source data, it would not be possible to link a synthesized item to other facts in a meaningful way. However, to the best of our knowledge, the assumption that synthetic datasets are immune to linkage attacks has not yet been formally proven. In Section~\ref{section3}, we supply such a proof; we rely on the uniqueness and reconstruction probability of the generated sequences and utilize the inverse relationship between the two to limit the existence of rare combinations and demonstrate that generating data using PrivGen provides privacy protection.

PrivGen is effective for privacy preserving data publication, as well as simulation, testing, and forecasting purposes. Unlike other perturbation methods, the suggested method preserves sequential patterns and various levels of time periodicity. It is generic and enables modeling different types of dimensions, including the time dimension, the state dimension which creates the sequence of events over time, and possibly other dimensions that characterize the objects in the data.

We make the following contributions in this paper: 1) Proposing PrivGen, a method for sequence generation, and demonstrating its ability to maintain the sequential patterns as well as various attribute statistics, while maintaining privacy, 2) demonstrating how using PrivGen provides protection against data privacy attacks, 3) providing a comparison between different machine learning methods for the task of extracting sequence patterns, and 4) conducting extensive experiments to demonstrate the superiority of PrivGen over state of the art techniques.

The rest of the paper is organized as follows: in Section~\ref{section2} related work is summarized. In Section~\ref{section3} preliminary definitions are provided, as well as proof that the method preserves privacy. The feature engineering performed is discussed in Section~\ref{section4}. Section~\ref{section5} describes the model, and Section~\ref{section6} presents the data generation algorithm. Section~\ref{section7} contains an analysis of the method's time complexity, and in Section~\ref{section8} we discuss experimental results. Finally, Section~\ref{section9} concludes this work.

\section{Related Work}
\label{section2}

Several methods have been suggested for publishing sequential data. Some proposed techniques \cite{chen2013privacy,mohammed2010preserving,Ghasemzadeh2014,Pensa2008,amiri2016hierarchical,gramaglia2017k,zhang2018towards,dong2018novel} adopt solutions based on the k-anonymity approach, such as LKC – anonymity, that make some adjustments to the k-anonymity approach to adapt to limitations deriving from the sequential nature of the data. Unfortunately, k-anonymity has been shown to be ineffective against many types of attacks \cite{machanavajjhala2006diversity,shokri2011quantifying}, and therefore most of the research performed has primarily aimed at publishing sequence data under differential privacy constraints \cite{dwork2018privacy,mcmahan2017learning,abadi2016deep, bindschaedler2017plausible,Shaked2016}, which is currently the most widely accepted formal mathematical model that ensures privacy protection \cite{yang2017survey}. Zhang et al.~\cite{zhang2017privbayes}, for example, present PrivBayes, a differential generative model that decomposes high dimensional data into low dimensional marginals by constructing a Bayesian network. Differential privacy is achieved by injecting noise into the low dimensional marginals and synthetic data is generated according to these noised marginals. 

Chen et al. propose two algorithms for publishing sanitized sequential data from which frequent sequence patterns can be mined. They build the synthetic datasets based on 1) a noisy prefix tree that combines sequences with identical prefixes for the same branch \cite{chen2012differentially_1}, and 2) a variable n-gram model \cite{Chen2012}, which can extract the essential information of the datasets in order to adjust the scales of injected noise. The PrivTree technique \cite{zhang2016privtree} uses a representation of a variable length Markov chain model to overcome the need to set some predefined threshold for the recursion depth, as required by the previously mentioned techniques.

Alongside these solutions, the use of deep architectures for training a generative model has been increasing in recent years. DP-SYN \cite{abay2018privacy,frigerio2019differentially} is a differentially private deep learning based synthetic data generation technique; it first partitions the original data into groups, and then employs the private auto-encoder for each group. The auto-encoder learns the latent structure of each group and uses expectation maximization algorithm to simulate the groups. The authors demonstrated how deep learning models have been trained to successfully capture relationship among multiple features, and then used to generate differentially private synthetic data. This method was shown to handle continuous, categorical and time-series data, but it does not support the problem that we aim to solve in this study, which is synthesizing multi-dimensional data with a sequence of discrete events. Choi et al.~\cite{choi2017generating} generated synthetic electronic health records (EHR) with a GAN framework. These EHRs contain binary and discrete variables, such as drugs and procedure codes. The authors defined an algorithm to generate high-dimensional multi-label (but non-sequential) discrete samples by combining an autoencoder with a GAN, which they call medGAN. However, medGAN preserves counting queries rather than sequences. Yu et al.~\cite{yu2017seqgan} developed SeqGAN for the generation of sequences of discrete data. SeqGAN is based on a reinforcement learning training where the reward signal is produced by the discriminator. They demonstrated its use for text and music generation. SeqGAN, however, handles one dimensional data instead of multi-dimensional data (both numeric and categorical) as we aim to support in this study. These recent methods do not support the complex nature of the data we wish to synthesize with PrivGen (multidimensional data that includes a sequence dimension), and therefore we can neither use them for the problem we strive to solve, nor compare them to PrivGen. Adapting each of these methods to such complex data is an interesting future area of research.

While the concept of differential privacy has received considerable attention, in reality, it is accompanied by quite a few difficulties; Lee et al. \cite{lee2011much} have shown that establishing the privacy budget ($\epsilon$) requires knowledge of the queries to be computed, the universe of data, and the subset of that universe to be queried. They also show that given a goal of controlling probabilistic disclosure of the presence of an individual, the proper $\epsilon$ varies depending on individual values, and that the presence of outliers influences the budget.

The risk of privacy violation exists when disclosure allows an adversary to calculate a high enough probability that the individual is in the dataset. With a sufficient privacy budget, differential privacy ensures privacy protection to the individual, limiting the risk of privacy violation. However, by applying a differential privacy technique that relies on strong privacy transformations, we run the risk of needlessly harming the data’s quality. As stressed in \cite{soria2017individual}, the technique of differential privacy may seriously and unjustifiably damage data, as the computation of the global sensitivity does not consider the actual data to be protected. Recent personalized differential privacy-based solutions try to reduce the amount of information loss by determining the minimum amount of transformations needed for satisfying the privacy requirements of each moving object. In the PPTD approach \cite{komishani2016pptd}, which also aim at reducing the amount of information loss, the disclosure risk is used as a metric to measure the privacy breach probability of moving objects and determine the minimum amount of necessary generalization and local suppression for satisfying the privacy requirements of each moving object.

Due to the aforementioned limitations of differential privacy, new approaches have recently been suggested that quantitatively try to manage privacy risks. These techniques also provide a kind of personalized privacy. Dankar and Emam \cite{dankar2010method} introduce a re-identification risk metric that measures the proportion of records that are correctly re-identified in a dataset by an adversary, based on the idea that the adversary wishes to re-identify as many records as possible in the published dataset. Tools like ARX \cite{prasser2016lightning,bild2018safepub} and PRUDEnce \cite{pratesi2018prudence} enable the analysis of the risk of re-identification and the application of various privacy paradigms. They define privacy risk as the probability that a specific user is re-identified in a specific published dataset, i.e., the probability that a specific user's identity is correctly associated with the user's own data, and also define empirical privacy risk as the distribution of the risk, i.e., re-identification probability, over the entire population of users represented in the dataset. These tools search among possible transformations to select safe transformations that minimize utility loss. The authors of \cite{trabelsi2009data} evaluate the personal data disclosure risk using an entropy-based method that enables the estimation of the rareness of some particular values within a population. Most of these techniques apply a privacy preserving transformation to eliminate the users with a risk higher than a certain threshold, ensuring that all users in the sanitized data fall below that risk. 

In a spirit similar to the one behind the personalized privacy concept, we take advantage of the fact that the probability of sequence reconstruction is inverse to the sequence's frequency within data and demonstrate that the use of generated sequence data ensures privacy. Existing perturbation-based techniques that aim at generating privacy preserving synthetic data are usually based on some condensing that preserves the statistical properties of the original data, like microaggregation \cite{aggarwal2008static,soria2018differentially} or prefix trees \cite{chen2012differentially_1,Chen2012}. To better capture sequence patterns in different periodicities, we examine condensing the original data into machine learning models for predicting both time and events according to various derived features extracted from both the time and sequence of events, as well as according to other features when available. 

The proposed PrivGen method bridges gaps such as the need to change the data resolution for generalization-based methods or the need to choose an appropriate privacy budget for differential privacy-based solutions, and it is only suitable for scenarios where individual sequences need not be truthful, and the emphasis is mainly on maintaining core statistics.
\section{Preliminaries and Disclosure Risk in Synthesized Data}
\label{section3}
Our input space $X$ is of the form $X= o, a_1,\ldots a_m,(t_1, s_1),$ $\ldots (t_i, s_i)$, where $o$ is the attached identifier of the sequence owner, $a_1,\ldots a_m$ are possible influencing attributes, and $(t_i, s_i)$ are time and state pairs that the owner is attached to. This also can be viewed as a sequence of states (or events) $S_j = s_1 s_2 \ldots s_i$ ordered by $t_i$. 

A sequential dataset $D$ of size $\left |D\right |$ is composed of a multiset of sequences $D = \{S_1,S_2,\ldots,S_{\left |D\right |}\}$. 

\begin{exmp}
\label{exmp:e1}
Let $X =$ \\
\{(1, '$M$', 45, 1/2/2017 5:45, '$Hospitalization$ $in$ $ER$'),\\
(1, '$M$', 45, 2/2/2017 15:00, '$Release$ $from$ $ER$'),\\
(2, '$F$', 45, 1/2/2017 6:40, '$Hospitalization$ $in$ $ER$')\} \\
be an input dataset. There are two influencing attributes in this example: the first is gender, and the second is age. The state relates to the type of transition between hospital units (emergency room hospitalization or release in this example). 
\end{exmp}

\begin{problem}[Publishing Anonymized Generated Data Problem]
Given an input dataset $D = \{S_1,S_2,\ldots,S_{\left |D\right |}\}$, create an output dataset ${D}' = \{{S}'_1,{S}'_2,\ldots,{S}'_{\left |{D}'\right |}\}$, where ${D}'$ preserves utility and provides anonymity as defined later in this section.
\end{problem}

We guard against identity disclosure, also known as a linkage attack, which occurs when the attacker is able to associate a record in the released dataset with the individual that it originated from and then can gain knowledge regarding additional confidential attributes.

\begin{defn}[attack model]
An attacker with some external knowledge might try to perform record linkage in order to extract more information about objects that appear in the dataset. To protect against identity linkage, we define a quasi-identifier (QID), which is the set of attributes that may be used for linking a record with external knowledge. In the simplified problem, we define the QID as $\{a_1,a_2,\ldots,a_m,s_1\}$, although it is easy to rely on the same principles to support more components or even an entire sequence $\{a_1,a_2,\ldots,a_m,t_1,s_1,t_2,s_2,\ldots\}$
\end{defn}

Measuring disclosure risk is based on the evaluation of the probability of correct re-identification of individuals in the released data. In general, the rarer a combination of values of the quasi-identifiers of an observation in the sample, the higher the risk of identity disclosure. An attacker that tries to match an individual who has a relatively rare quasi-identifier within the released data with an external dataset in which the same quasi-identifier exists will have a higher probability of finding a correct match than when a larger number of individuals share the same QID \cite{benschop2018statistical}. We use two main concepts to measure identity disclosure risk: uniqueness and reconstruction probability. Uniqueness relates to how rare combinations of feature values are in the original dataset, and reconstruction probability estimates the probability that combinations of feature values in the original data will appear in the released data. Using the inverse relationship between the uniqueness of a combination and its reconstruction probability, we demonstrate a low probability of rare combinations in the released data. This makes the proposed method immune to linkage attacks and useful for protecting the privacy of individuals.

\begin{defn}[reconstruction probability]
We define reconstruction probability as the probability that a given QID from the source data ($D$) appears in the generated data (${D}'$).

Let's define $Y$ as follows:
\begin{equation}
y_i =\begin{cases}1& QID(x_i) \subseteq  QID({D}')\\0& otherwise\end{cases}  
\end{equation}
Assuming that influencing attributes are independent, and the state dimension depends on all other features, the distribution across $X \times Y$ is:
\begin{equation}
\label{eq:eq2}
\begin{split}
& P(y_i=1 | x_i ) = \\
& p(x_i.{a_1} )\cdot\ldots\cdot p(x_i.{a_m} )\cdot p( x_i.s | x_i.{a_1},\ldots,x_i.{a_m})
\end{split}
\end{equation}
\end{defn}

Each element in this calculation refers to the probability of generating a certain value for one dimension of the QID, where the multiplication calculates the probability that all attributes have been successfully reconstructed to the values of a given record.

The reconstruction probability for a record $x_i$ is $P(y_i=1|x_i)$.

\begin{prop}[reconstruction probability]
Let's denote our suggested algorithm as $A$. Algorithm $A$ is based on estimating $p(x_i.a_1 )\cdot\ldots(x_i.a_m )$ according to the distribution of the possible discretized values of the attribute in a given dataset $D$ and an additional model for each dimension. In the simplified model we use for this proof there is an additional dimension for $s_1$, estimated by algorithm $A$ using a model $c_1$ with a generalization error $\epsilon_1$.

The reconstruction probability is therefore estimated as:
\begin{equation}
\label{eq:eq3}
\begin{aligned}
& P_A (y_i=1 | x_i )=\\
& p(x_i.{a_1} )\cdot \ldots \cdot p(x_i.{a_m} )\cdot c_s( x_i.{s}|x_i.{a_1}, \ldots ,x_i.{a_m} )=\\
& p(x_i.{a_1} ) \cdot \ldots \cdot p(x_i.{a_m} ) \cdot p( x_i.s  |  x_i.{a_1}, \ldots ,x_i.{a_m} )\\
& \cdot(1- {\epsilon}_1 ) 
\end{aligned}
\end{equation}
\end{prop}

\begin{defn}[uniqueness]
We define uniqueness as the opposite of the frequency of $QID(x_i)$ in $D$. Since the frequency of $QID(x_i)$ is $P(y_i=1|x_i )$, we use the complementary frequency $1-P(y_i=1|x_i )$ for calculating uniqueness. The higher the uniqueness value, the easier it will be for an attacker to link external data to the record owner, as less owners share this QID.
\end{defn}
\begin{defn}[re-identification Risk]
\label{def:identificationrisk}
We define the re-identification risk as the probability that a given QID will be successfully used for re-identifying a generated record. This probability is composed of the probability that a QID will appear in ${D}'$ and its frequency within the data, as rare QIDs will be linked to less possible owners. 
\begin{equation}
\label{eq:identificationrisk}
\begin{aligned}
&ReIdentificationRisk=\\ 
&reconstructionProbability\cdot uniqueness=\\
&P(y_i=1 | x_i )\cdot (1-P(y_i=1 | x_i ))
\end{aligned}
\end{equation}
\end{defn}
\begin{exmp}
Let $X =$ \\
\{(1, '$M$', 45, 1/2/2017 5:45, '$Hospitalization$ $in$ $ER$'),\\
(2, '$M$', 45, 2/2/2017 15:03, '$Hospitalization$ $in$ $ER$'),\\
(3, '$M$', 45, 3/2/2017 12:23, '$Hospitalization$ $in$ $ER$')\}.

The reconstruction probability for record $x_1$ is $P(y_1=1 | x_1 )$= $p('M')\cdot p($45$)\cdot$ p$($'Hospitalization in ER'$ |$'M'$,45)$= $3/3 \cdot 3/3 \cdot3/3$ = 1. We obtain the following: $reconstructionProbability=1$ and $uniqueness=0$ for $x_1$. Although record $x_1$ was successfully reconstructed in ${D}'$, it has $n$ possible sources, and an attacker could not distinguish between $n$ possible owners. The re-identification risk for record $x_1$ is $1\cdot 0=0$.

The reconstruction probability for record (4, '$M$', 45, 1/2/2017 5:45, '$Release from ER$') is $P(y_1=1 | x_1 )$=$p('M')\cdot p(45)\cdot p($'Release from ER'$ | $'M'$,45)$=$3/3\cdot 3/3\cdot 0/3$= 0, as the probability for at least one of its attributes is zero. That means that this record does not appear in ${D}'$, so it may not be used for identifying an owner. The re-identification risk in this case is zero.
\end{exmp}

\begin{exmp}
Let $X$ =\\
\{(1, '$M$', 45, 1/2/2017 5:45, '$Hospitalization$ $in$ $ER$'),\\
(1, '$M$', 45, 2/2/2017 15:03, '$Hospitalization$ $in$ $ER$'),\\
(2, '$M$', 45, 2/2/2017 15:03, '$Hospitalization$ $in$ $ER$'), \\
(2, '$M$', 45, 3/2/2017 12:23, '$Release$ $from$ $ER$'),\\
(3, '$F$', 16, 2/2/2017 17:03, '$Hospitalization$ $in$ $ER$'),\\
(3, '$F$', 16, 3/2/2017 12:56, '$Release$ $from$ $ER$')\}.
The reconstruction probability for record $x_1$ is $P(y_1=1 | x_1 )$=$p('M')\cdot (45)\cdot p('Hospitalization$ $in$ $ER'|'M',45)$= $2/3\cdot 2/3\cdot 2/2$= 0.44. Note that influencing attributes are common to all records of an owner, and therefore their probability is calculated according to the number of owners instead of the number of records. We obtain the following: $reconstruction-$ $Probability = 0.44$ and $uniqueness =0.56$ for $x_1$, and its re-identification risk is $0.44\cdot 0.56=0.2464$.
\end{exmp}

\begin{thm}[maximal re-identification risk]
\label{thm:maxIdentificationrisk}
Let $M(D)$ be a data generation method that was trained based on a source dataset $D$ and provides reconstruction probability p. Let ${D}'$ be a dataset generated using M(D). Then, the re-identification risk when publishing ${D}'$ is 0.25 or less.
\end{thm}
\begin{proof}
Based on Definition~\ref{def:identificationrisk}, the re-identification risk is measured as $p\cdot (1-p)$, where $p$ denotes the reconstruction probability. Fig.~\ref{fig:Fig1} displays the re-identification risk as a function of $p$. Since $0\leq p\leq 1$, the result of multiplying $p$ by $(1-p)$ ranges from zero to 0.25. Hence, it follows that for any record $x_i$ in $D$, $0\leq ReIdentificationRisk\leq 0.25$. For a dataset $D$ with $n$ records, $0\leq ReIdentificationRisk\leq(0.25)^n$.
\end{proof}
\begin{figure}[htbp!] 
\centering    
\includegraphics[width=\linewidth]{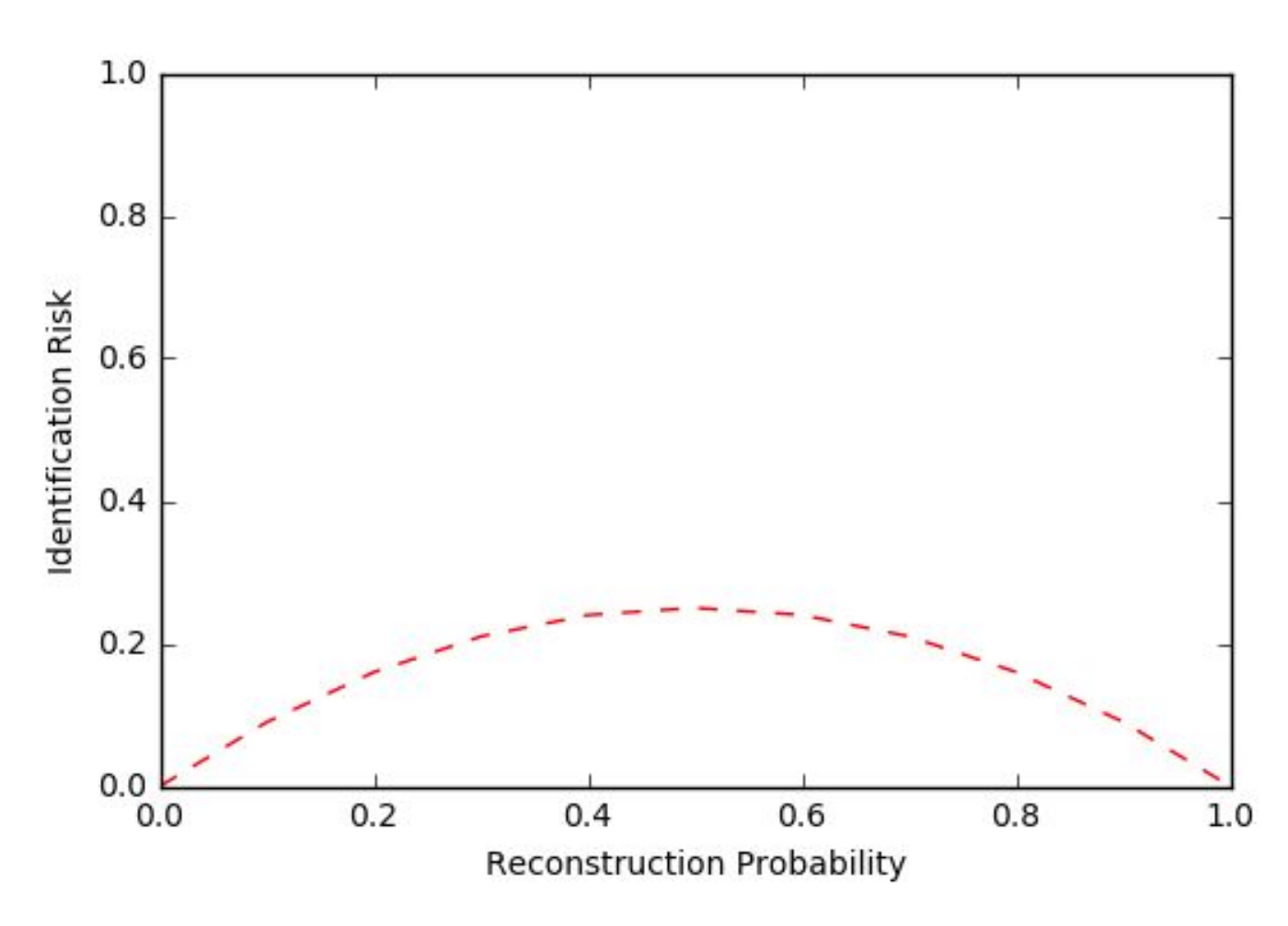}
\caption[Re-identification risk versus reconstruction probability.]{Re-identification risk versus reconstruction probability.}
\label{fig:Fig1}
\end{figure}

Theorem~\ref{thm:maxIdentificationrisk} guarantees a low identification risk when publishing synthesized data.

Fig.~\ref{fig:fig2} shows the re-identification risk versus the reconstruction probability and the error $(\epsilon_s)$, according to Eq.~\ref{eq:identificationrisk}, where reconstructionProbability is estimated by $(1-\epsilon_s) \cdot reconstructionProbability$. The reconstruction probability clearly dictates the identification risk, which is maximized for reconstructionProbability = 0.5. The error ($\epsilon_s$) has a slight influence on the identification risk, as low error provides a minor increase in the re-identification risk.

\begin{figure}[htbp!] 
\centering    
\includegraphics[width=\linewidth]{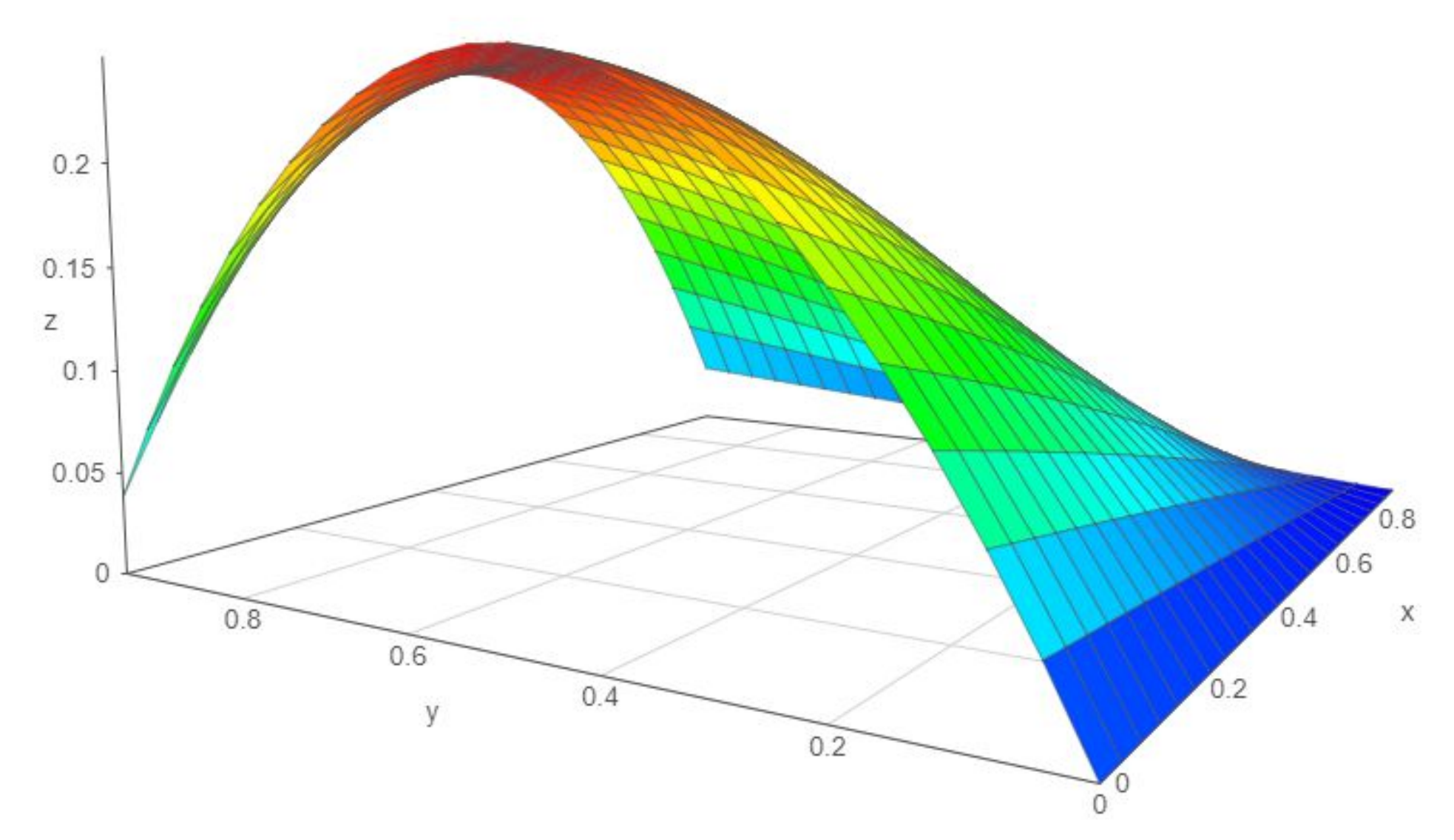}
\caption[Identification risk vs. reconstruction probability and $\epsilon_s$]{Re-identification risk ($z$) vs. reconstruction probability ($y$) and $\epsilon_s$ ($x$).}
\label{fig:fig2}
\end{figure}

As with any other data generation-based solution, the utility of the released data will not be measured at the record level but rather in the ability to produce data with characteristics similar to those of the original data, with an emphasis on similar sequence patterns and time periodicities. 

\begin{defn}[data utility]
\label{def:datautility}
Utility is estimated according to the similarity between the following tasks based on the source and released data: 1) accuracy of predicting next state based on a model trained on the original data; 2) comparing distributions for various attributes; 3) precision in identifying the top-k frequent subsequences in $D$.

As stated before, we simplified the problem to provide a simpler proof, so that the quasi-identifier only includes $\{a_1,a_2,\ldots$ $a_m,s_1\}$, but this can easily be extended to $\{a_1,a_2,\ldots a_m,t_1,s_1,$ $t_2,s_2,\ldots t_i,s_i\}$ by using a state transition model to predict the next state given the previous states and by using a time transition model to predict the next timestamp given previous times and states. This change will add additional probabilities to the product in Eq.~\ref{eq:eq3} but will not change the next steps in the proof.
\end{defn}

\section{Data Preprocessing and Feature Engineering}
\label{section4}
The PrivGen methodology consists of three main phases, as illustrated in Fig.~\ref{fig:GraphicalAbstract}: 1) Feature engineering and preprocessing, 2) Data modeling, and 3) Data generation. In this section, we discuss the feature engineering and preprocessing phases. To enhance the model, we extract additional features from the original features. The new features are derived from the time dimension and the sequential nature of the data. The new features allow the model to detect cyclicality at different levels, for example, identifying a daily pattern next to a weekly pattern. The complete list of engineered features is presented in Table~\ref{tab:engineeredfeatures}.
\begin{figure*}[htbp!] 
\centering    
\includegraphics[width=\linewidth]{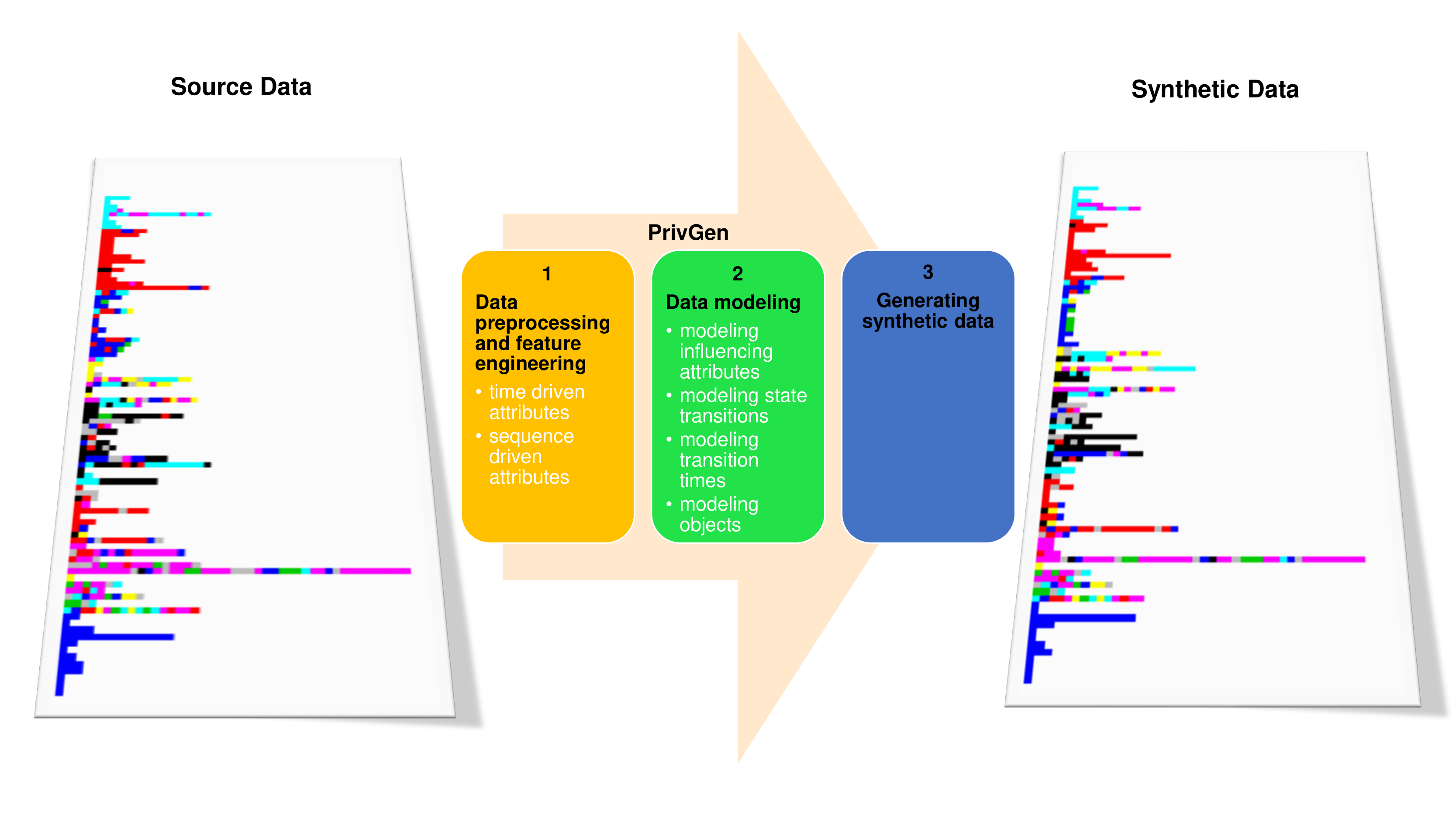}
\caption[PrivGen methodology]{PrivGen methodology.}
\label{fig:GraphicalAbstract}
\end{figure*}
\begin{table}[!htb]
\caption{Engineered features.}
\label{tab:engineeredfeatures}
\centering
\begin{tabular}{ l  p{5.5cm} } 
\toprule
\multicolumn{2}{l}{\textbf{Time Driven Attributes}}\\ 
\cmidrule{1-2}
\multicolumn{1}{l}{Feature} & \multicolumn{1}{l}{Description}  \\ 
\midrule
weekday	&day of week	\\
hour	&hour in day	\\
month	&month\\
week of year	&number of the week in the year\\
quarter	&quarter in the year\\
year	&year\\
\midrule
\multicolumn{2}{l}{\textbf{Sequence Driven Attributes}}\\ 
\cmidrule{1-2}
\multicolumn{1}{l}{Feature} & \multicolumn{1}{l}{Description}  \\ 
\midrule
seq start	&1 for the first state in a sequence, 0 otherwise\\	
seq ID	&a unique running number ID for each sequence\\	
tranTime &time (in seconds) from current state to the next state in the sequence (target)\\
cumTime	&time (in seconds) that passed from the beginning of the sequence to the current state\\
state order	&a running number that states the order of the state within the sequence\\	
$tran0,1,\ldots|W|$	&$|W|$ features for the $|W|$ previous tranTime in the sequence ($|W|$ is the sliding window size)\\
next state	&the next state in the sequence (target)\\
prev w $0,1,\ldots|E|$&$|E|\cdot |W|$ features will be extracted, one per state that appeared $w$ steps ago, to enable a one-hot representation\\
lastState$0,1,\ldots|E|$ &$|E|$ features will be extracted, one per state; each score reflects the number of occurrences of the state within the window ($|W|$) that slides in previous states \\
\bottomrule
\end{tabular}
\end{table}

\subsection[Time Driven Attributes]{Time Driven Attributes}
The time dimension contains a large amount of information that can be exploited to better understand the sequence patterns of state transitions. As an example, consider the movement patterns of an object. The hours of the day, and even the day of the week, influence the nature of the movement patterns (sleep with little movement at night, commute to work on weekday mornings, travel to less familiar places on weekends, etc.). We therefore extract a set of features from the time dimension.
\subsection[Sequence Driven Attributes]{Sequence Driven Attributes}
The sequence structure encapsulates important information that can be exploited to improve the model. The prediction of the model for each step is influenced by sequence-based features; the first event may behave differently than the fifteenth event, the time that has passed from the beginning of the sequence may influence the course of events, and so on. In addition, we derive sliding window-based features. A size five window, for example, means extracting characteristics about the previous five states in the sequence. Also, note that when extracting the target attribute (next state), we use an additional state to indicate the end of the sequence. In this way, we can predict the end of the sequence.
\subsection[Data Preprocessing]{Data Preprocessing}
Nominal features that are assumed to influence other features are discretized into a predetermined number of bins. Eventually, all nominal features are converted into binary features (in a one-hot representation) to accommodate the widest possible range of machine learning algorithms.
\section{Data Modeling}
\label{section5}
The second phase in PrivGen, which is performed upon completion of feature engineering (phase one), is data modeling. PrivGen’s model combines several components to identify characteristics of different dimensions of the data (a component to model transition times, a component to model transition between states, etc.). To support such an architecture, we must make assumptions about the existing dependencies between the different dimensions, as modeling independent dimensions differs from modeling dependent dimensions. The assumed dependencies between the various dimensions are presented in Fig.~\ref{fig:fig3}. The features derived are new dimensions extracted based on the original dimensions (time, state, and object ID): time driven features are extracted from the time dimension, and sequence driven features are extracted by rearranging the data based on the object, time, and state dimensions. The state dimension depends on the influencing attributes, the time dimension, and all of the other features extracted. The time dimension depends on all other dimensions. The object dimension is independent and does not influence any other dimension. The influencing attributes are independent but influence all of the dimensions.
\begin{figure}[htbp!] 
\centering    
\includegraphics[width=\linewidth]{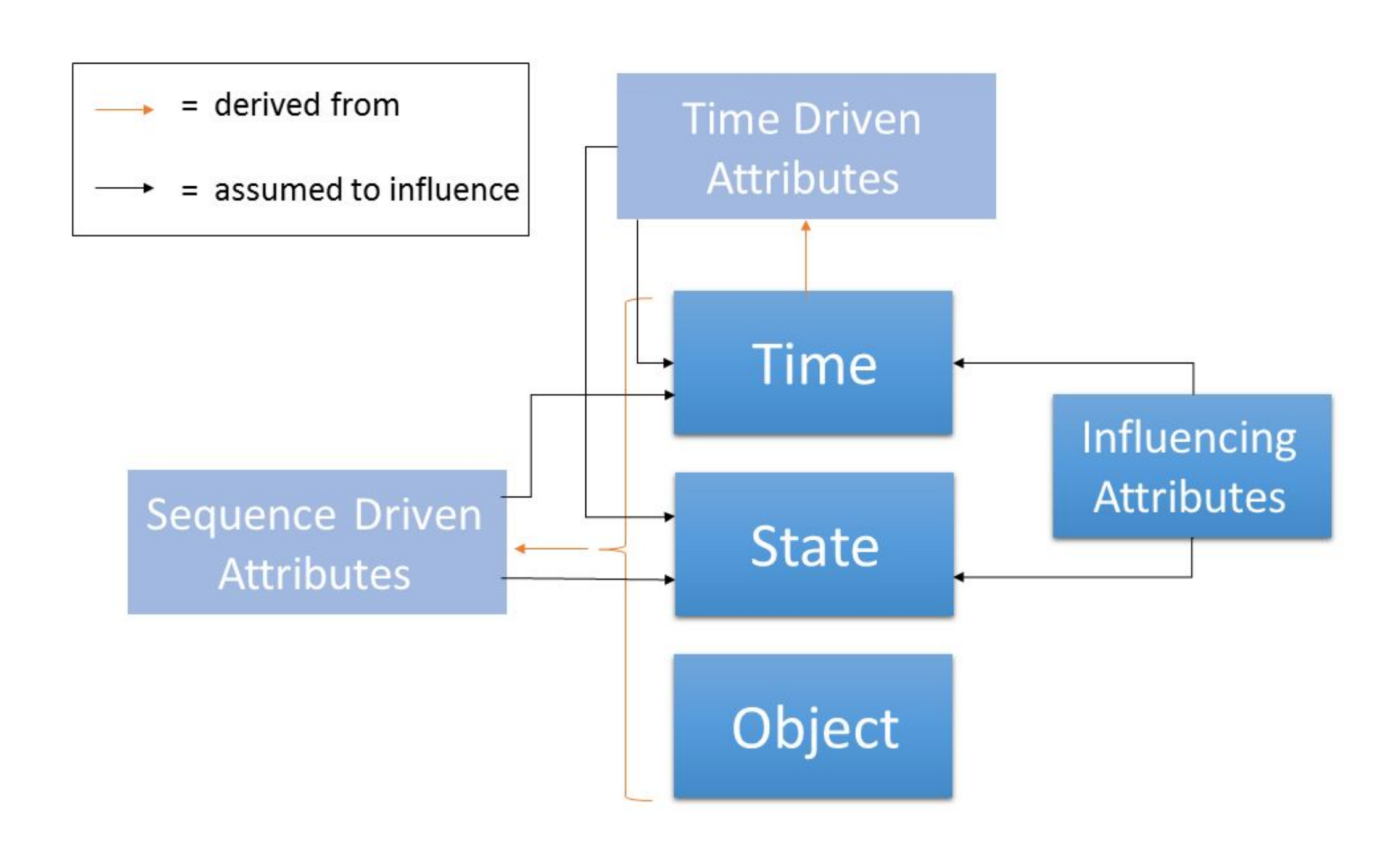}
\caption[Dependencies between data dimensions]{Dependencies between data dimensions.}
\label{fig:fig3}
\end{figure}

\subsection[Modeling Influencing Attributes]{Modeling Influencing Attributes}
Influencing attributes are features that are attached to objects and may affect sequence patterns. Age or gender, for example, may affect the pattern of item acquisition in a supermarket or the pattern of medical symptoms experienced by a patient.

Since these attributes are assumed to be independent, their modeling is simply based on extracting their distribution in the source dataset. The frequency of each possible value is calculated for nominal attributes. A preliminary stage of discretization is performed for numerical features. The current discretization is set to 100 bins, but it can be set differently to support various features.
\subsection[Modeling State Dimension]{Modeling State Dimension}
We model state transitions as a multiclass classification task; a model is trained to predict each state's probability of being the next state. This creates an ensemble of models (one model for each state). Voting among predictions of the ensemble members by choosing the maximal probability prediction yields the prediction of the next state. The input for this model includes all of the attributes (except object ID), as well as all other features extracted. We examine four machine learning algorithms for this task: Decision Tree, Random Forest, XGBoost, and Markov model.
\subsection[Modeling Time Dimension]{Modeling Time Dimension}
Transition times are modeled as a prediction task. The input for this model also includes all of the attributes (except object ID) and the other features extracted. The subsequent state feature extracted also serves as input, as it is assumed to affect the transition times (assuming that transitions between two given states have similar durations). After an initial examination of several algorithms for this task, we chose to use the Random Forest algorithm that performed best with our configurations, although other algorithms can be used as well.
\subsection[Modeling Object Dimension]{Modeling Object Dimension}
The object dimension is not modeled, as new objects will be generated with running numbers as IDs. The only step required here is to store the number of objects in the source data, as we wish to preserve the number of objects in the generated data.

\section{Generating Synthetic Data}
\label{section6}
Once each model component has been trained, the various models can be used to generate new data of a similar nature. This is where the last phase (data generation) of PrivGen fits in.

The data generation algorithm is outlined below (Algorithm ~\ref{algo:datageneration}). The input to the data generation algorithm is a trained model ($m$), a sliding window size ($w$), and a number of objects to generate ($n$). The last parameter is the maximal number of iterations (\textit{max-steps}) which determines the maximal number of states in a sequence.

\IncMargin{1em}
\begin{algorithm}[!htb]
\SetKwData{tmp}{tmp}\SetKwData{maxSteps}{max-steps}
\SetKwFunction{GenerateUniqueID}{GenerateUniqueID}\SetKwFunction{SampleStartTime}{SampleStartTime}\SetKwFunction{SampleAttrI}{SampleAttrI}\SetKwFunction{SampleStartState}{SampleStartState}
\SetKwFunction{SampleAttrM}{SampleAttrM}\SetKwFunction{PreProcess}{PreProcess}
\SetKwFunction{GetLastByObject}{GetLastByObject}\SetKwFunction{PredictState}{PredictState}
\SetKwFunction{PredictTranTime}{PredictTranTime}\SetKwFunction{GetLastByWindow}{GetLastByWindow}
\SetKwInOut{Input}{input}\SetKwInOut{Output}{output}
\Input{A trained model $m$,\ the number of objects to be generated $n$,\ sliding window size $w$,\ and maximal number of iterations \maxSteps}
\Output{The generated dataset ${D}'$}
\BlankLine
${D}'\leftarrow$ create $n$ records of the form $(o, t, e, a_1, \ldots, a_m)$\;
${{D}'[o]}\leftarrow$  \GenerateUniqueID()\;
${{D}'[t]}\leftarrow$  \SampleStartTime($m,{D}'$)\;
${{D}'[e]}\leftarrow$  \SampleStartState($m,{D}'$)\;
${{D}'[a_1]}\leftarrow$ \SampleAttrI($m,{D}'$)\; 
$ \vdots$\\
${{D}'[a_n]}\leftarrow$ \SampleAttrM($m,{D}'$)\;
\Repeat{\tmp is empty or \maxSteps has been reached}{
  \tmp  $\leftarrow$ \GetLastByWindow(${D}',w$); \tcp{get $w$ last records of sequences that had not yet reached the ending state}
  \tmp  $\leftarrow$ \PreProcess(\tmp); \\
  {\tmp}$_e$  $\leftarrow$ \GetLastByObject($m$,\tmp); \tcp{get last record of each sequence in tmp}
  {\tmp}[e]  $\leftarrow$ \PredictState($m$,\tmp)\;
  {\tmp}[t]  $\leftarrow$ {\tmp}[t] + \PredictTranTime($m$,\tmp)\;
  Append \tmp to ${D}'$\;
}
\Return ${D}'$
\BlankLine
\caption{Data generation algorithm.}\label{algo:datageneration}
\end{algorithm}
\DecMargin{1em}

First, the algorithm initializes the generated dataset ${D}'$ by creating $n$ records, generating a unique object ID for each record and populating the records with sampled start times, start states, and values for each of its $m$ attributes. It then repeats, sampling the next step (record) for as long as there are sequences with this number of steps (or until meeting the \textit{max-steps} limit). In each iteration, the algorithm stores (in $tmp$) the $w$ last steps of sequences that have not yet reached the ending state. The $w$ steps are needed for the extraction of all relevant features, as performed in the following preprocessing phase in which the last preprocessed record (step) for each object (sequence) in $tmp$ is used for predicting the next state and transition time. $tmp$ is updated with the predicted times and states, and is added to ${D}'$ which will eventually be returned as output.

\section[Time Complexity Analysis]{Time Complexity Analysis}
\label{section7}
The time complexity for extracting the required models for the suggested method is as follows: 1) Modeling each influencing attribute by scanning its $n$ values is $O(n)$, where $n$ is the number of records. 2) Modeling transition times using the Random Forest algorithm requires constructing $k$ trees. The complexity of tree construction is $O(v\cdot n \cdot log(n))$, where $n$ is the number of records, and $v$ is the number of attributes; we assume the worst case of a full depth tree. The input for this model contains all of the attributes, including an attribute for each state, so the complexity is actually $O(m\cdot({v}'+|E|)\cdot n\cdot log(n))$, where $m$ is the percentage of attributes to be considered per tree, ${v}'$ is the number of attributes in the original input, and $|E|$ is the number of states. Therefore, for $t$ trees in the forest we obtain a complexity of $O(t\cdot m\cdot ({v}'+|E|)\cdot n\cdot log(n))$. 3) The complexity of state transition modeling depends on the algorithm chosen for the task. Since a model is created for each state, the number of states ($|E|$) affects the time complexity. For example, in the Random Forest algorithm, when $t$ trees are created per model, the complexity is $O(|E|\cdot t\cdot v\cdot n\cdot log(n))$ for $|E|$ models, each with $t$ trees, based on $v$ attributes, where $v=m\cdot ({v}'+|E|)$. In the XGBoost algorithm, the number of trees $t$ will be replaced with the number of iterations $k$. In a Decision Tree algorithm, $t$ is removed from the computation, as only a single tree is created for each state. A Markov model can be extracted by scanning transitions in all records and calculating the frequencies of each possible transition. Therefore, its complexity is $O(n\cdot |{\rho}|)$, where $\rho$ is the number of unique transitions in the data, and $\rho \leq {|E|}^2$. 

The time complexity for the data generation algorithm is influenced by each of these steps: 1) Generating $o$ records with $a$ influencing attributes is $O(o\cdot a)$, where $o$ is the number of objects. 2) Predicting the next time for each step of $o$ sequences, using the Random Forest algorithm, is $O(o\cdot t\cdot log(n))$ for searching $t$ trees per prediction. 3) Predicting the next state for each step for $o$ objects, using the Random Forest algorithm with $t$ trees (for example, when we need to predict a score for each of the $|E|$ states), is $(o\cdot |E|\cdot t\cdot log(n))$. Steps 2 and 3 are repeated for $p$ iterations (each iteration is a step in the sequence), so the overall data generation complexity is $O(p\cdot o\cdot |E|\cdot t\cdot log(n))$.

\section[Experimental Evaluation]{Experimental Evaluation}
\label{section8}
In Section~\ref{section3}, we demonstrated how PrivGen protects privacy. In view of the trade-off between maintaining privacy and utility, in this section we evaluate the data utility achieved by our suggested method. 

Our experiments are divided into two parts. In the first set of experiments we evaluated PrivGen, by comparing different algorithms for the task of modeling transitions between states and evaluating how well various characteristics of the source data were maintained.
In the second set of experiments, we compared PrivGen’s results with the results of the most similar techniques for publishing privacy preserving sequences or privacy preserving sequence pattern extraction.
\subsection{Evaluation Methods}
As our focus in this study is preserving the sequential nature of the synthetic data, we evaluate quality using sequence-based tasks. We tune and evaluate the method according to its accuracy in the task of state transition classification as well as according to its RMSE in the task of time transition prediction. Moreover, we measure the differences between accuracy / RMSE achieved when using a source-data trained model and a synthetic-data trained model and applying each on the source data. The idea behind this metric is that low differences between results of a model that was trained on the source data and results of a model that was trained on the synthetic data, indicate the effectiveness of the synthetic data.
The last evaluation we use for measuring the quality of the sequential nature of the synthetic data, is by measuring the precision of matching top-k frequent subsequences in the source and synthetic data, as detailed in Subsection~\ref{section8-2}. This measure enables the comparison of PrivGen with state of the art techniques as it was used to evaluate the other methods. 

\subsection{Evaluating Quality of the Data Generated}
We applied PrivGen to two real-world datasets; the first was taken from the MIMIC-III
(Medical Information Mart for Intensive Care III) database~\cite{Lee2011}, which is a large, freely available database comprised of de-identified health-related data associated with over 40,000 patients who stayed in Beth Israel Deaconess Medical Center’s critical care units between 2001 and 2012. We used the $TRANSFERS$ table which contains physical locations for patients throughout their hospital stay. The second dataset, Taxis, contains a complete year (from 1/07/2013 to 30/06/2014) of trajectories for all of the 442 taxis running in Porto, Portugal~\cite{moreira2013predicting}. We initially reconstructed the dataset as a collection of signals rather than a collection of paths. We then filtered it to an area of about 50 square kilometers ($8.85\times5.5$), and each coordinate pair was discretized to a system of grids measuring about one square kilometer. Finally, we reduced the transmission rate in our sample to one per minute, in addition to the first and last signals of each trip. Three data samples of different sizes (25,000, 100,000, and 250,000 records) taken from the two datasets were used to evaluate PrivGen. Table~\ref{tab:data1} describes the samples.

\begin{table*}[!htb]
\caption{Data description for the first set of experiments.}
\label{tab:data1}
\centering
\begin{tabular}{llll} 
\toprule
&\multicolumn{3}{c}{Taxis}\\ 
\cmidrule{2-4}
Records&\multicolumn{1}{l}{25K} & \multicolumn{1}{l}{100K} & \multicolumn{1}{l}{250K} \\ 
\midrule
Number of sequences&\multicolumn{1}{l}{2410}&\multicolumn{1}{l}{9779}&\multicolumn{1}{l}{24488}\\
Number of unique states& \multicolumn{1}{l}{48}&\multicolumn{1}{l}{48}&\multicolumn{1}{l}{48}\\
min date& \multicolumn{1}{l}{2013-7-1 3:00}  & \multicolumn{1}{l}{2013-7-1 03:00} & \multicolumn{1}{l}{2013-7-1 3:00}\\
max date  & \multicolumn{1}{l}{2013-7-1 17:36}   & \multicolumn{1}{l}{2013-7-3 12:44}  & \multicolumn{1}{l}{2013-7-6 8:27} \\
sequence size (min $\pm$ std) & $9 \pm 7$& $9 \pm 7$& $9 \pm 7$\\
sequence duration (sec)& $576 \pm 530$& $564 \pm 508$& $566 \pm 551$\\
transfer time (sec) & $55 \pm 70$&$55 \pm 64$& $55 \pm 72$\\
\midrule
&\multicolumn{3}{c}{MIMIC-III}\\ 
\cmidrule{2-4}
Records&\multicolumn{1}{l}{25K} & \multicolumn{1}{l}{100K} & \multicolumn{1}{l}{250K} \\ 
\midrule
Number of sequences&\multicolumn{1}{l}{4370}&\multicolumn{1}{l}{17418}&\multicolumn{1}{l}{44361}\\
Number of unique states& \multicolumn{1}{l}{17}&\multicolumn{1}{l}{17}&\multicolumn{1}{l}{17}\\
min date& \multicolumn{1}{l}{2100-6-28 19:31} & \multicolumn{1}{l}{2100-6-14 04:56}  & \multicolumn{1}{l}{2100-6-7 20:00}\\
max date  & \multicolumn{1}{l}{2209-2-14 21:59}  & \multicolumn{1}{l}{2209-2-14 21:59}   & \multicolumn{1}{l}{2210-8-24 19:43} \\
sequence size (min $\pm$ std)&  $5 \pm 5$& $5 \pm 5$& $5 \pm 4$\\
sequence duration (sec)& $3523 \pm 11573$& $3712 \pm 12027$& $2963 \pm 10360$\\
transfer time (sec) & $616 \pm 4403$& $647 \pm 4503$& $526 \pm 3934$\\
\bottomrule
\end{tabular}
\end{table*}

In order to comprehensively evaluate PrivGen, we examined the quality of the data generated. We applied the method to samples of our two datasets. Five different state transition models were compared (random selection, Markov model, CART, Random Forest, and XGBoost). The first algorithm, random selection, is included in order to demonstrate any improvement observed over random results. The methods described were implemented in Python (code is available at \url{https://github.com/sigal-shaked/SequenceAnonymizer}). We used 10-fold cross-validation for the evaluation, training a model for each algorithm based on the training set and generating 1000 anonymized sequences accordingly. We then trained a model for each algorithm based on the anonymized data and compared it to the previous model by evaluating its performance with the source test set.

\subsubsection{Maintaining State Transitions}
Figure~\ref{fig:fig6} presents the accuracy of the state transition model for each of the five algorithms compared. The blue bars represent the results achieved for the source test set based on a model trained on the source training set, and the green bars represent the results achieved for the source test set based on a model trained on the generated anonymized data. Therefore, the green bars should be similar to the blue bars to indicate that state transition patterns were preserved in the anonymized version of the data. As shown in the figure, basing the model on the anonymized version of the data reduces the accuracy by $0\%$ to $20\%$, which leaves us with about $50\%$ accuracy for the examined models; this is twice as much as that obtained in a random sample from the MIMIC-III data and 10 times better that obtained in a random sample from the Taxis dataset.

However, as can be seen in Fig.~\ref{fig:fig6}, there is no significant advantage in using machine learning algorithms that have been enriched with broad input features compared to the classic Markov model, at least for the tested datasets. Note that our tested datasets do not contain features regarding the object itself, but only intrinsic features that can be extracted from the time dimension and from the sequence of states. For this scenario the Markov model seems to be the best option.

\begin{figure*}[htbp!] 
\centering    
\includegraphics[width=\linewidth ]{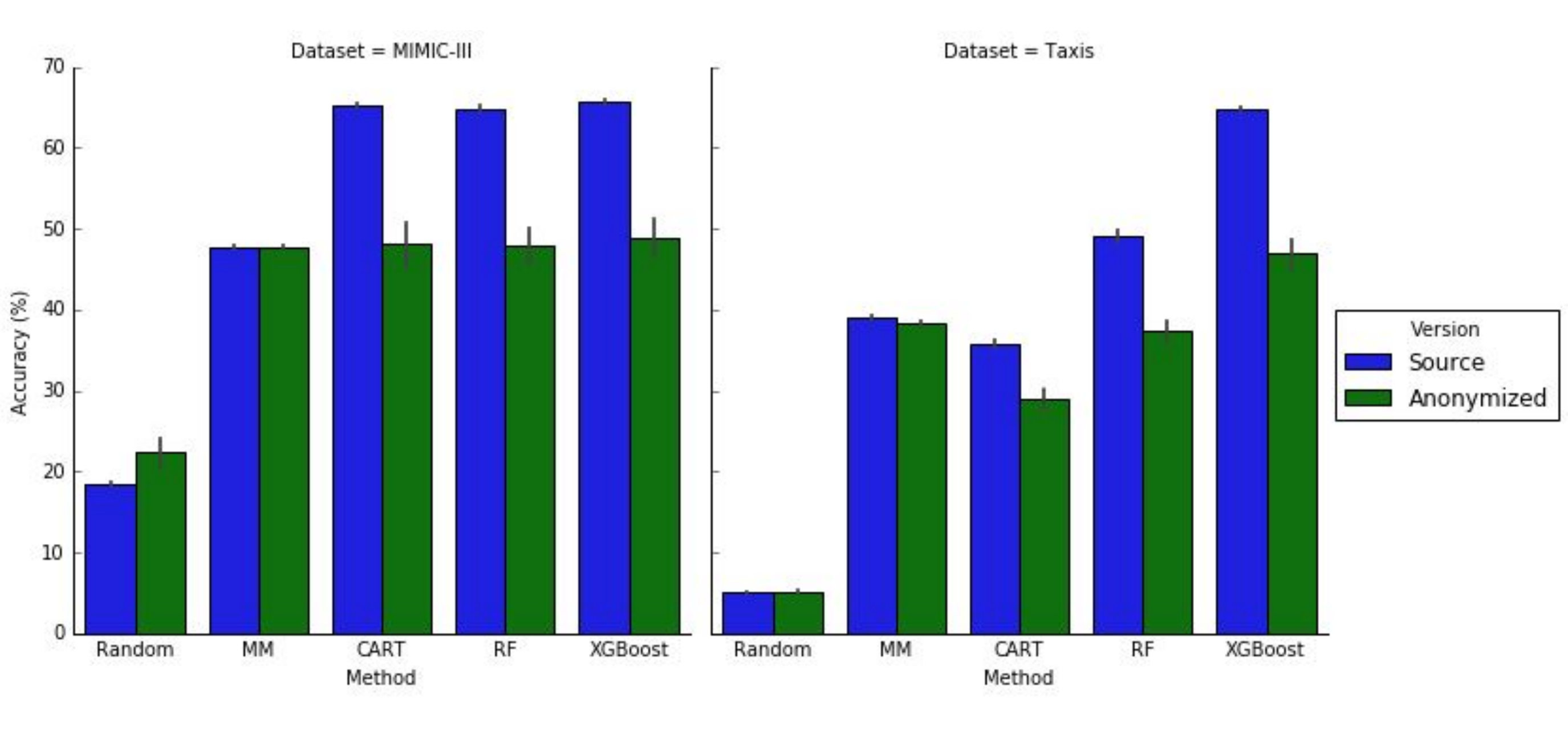}
\caption[Accuracies of state transition models, trained on source vs. anonymized data]{Accuracy of the state transition model using the following five algorithms:  Random Selection, Markov model (MM), Decision Tree (CART), Random Forest (RF), and XGBoost. The two colors distinguish between training a model on the source data (blue) and training a model on the anonymized data (green).}
\label{fig:fig6}
\end{figure*}
\subsubsection{Maintaining Time Transitions}
Fig.~\ref{fig:fig7} presents the RMSE of the time transition model. Here again, the blue bars represent the results achieved based on a model trained on the source training set, and the green bars represent the results based on a model trained on the generated anonymized data. A similar result for the source and anonymized models indicates the preservation of the patterns of the transition times. Indeed, we can see a similar amount of error in both cases, suggesting the preservation of transition time patterns using the anonymized data.
\begin{figure}[htbp!] 
\centering    
\includegraphics[width=\linewidth ]{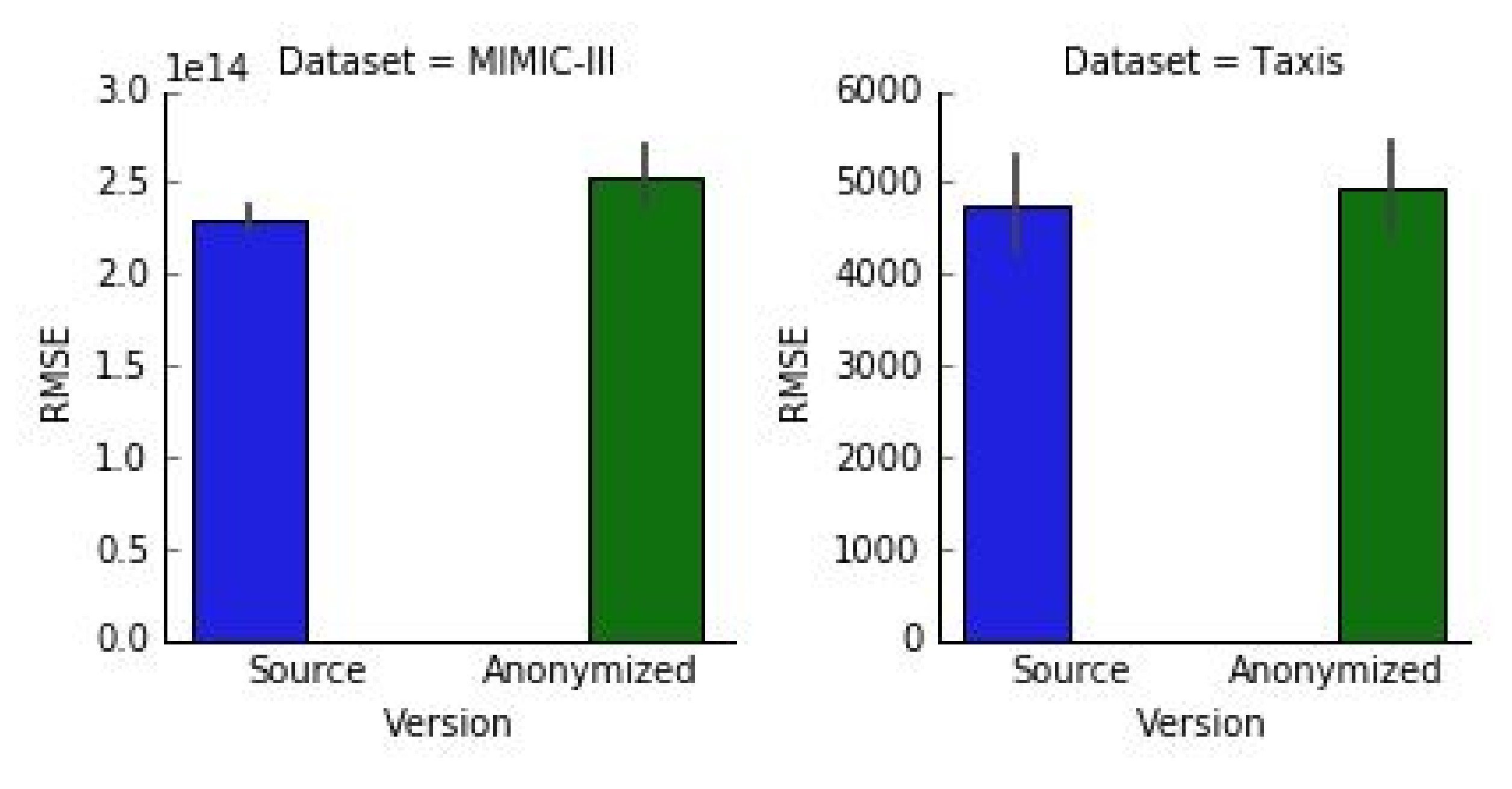}
\caption[RMSE of the transition times model]{RMSE of the transition times model. The two colors distinguish between training a model on the source data (blue) and training a model on the anonymized data (green).}
\label{fig:fig7}
\end{figure}

\subsubsection{Maintaining General statistics}
We examined the preservation of statistics that are relevant to the time dimension; the average sequence length is measured by the number of records associated with each sequence, and the average sequence duration is measured by the sum of the transition times of records associated with a sequence. As can be observed in Fig.~\ref{fig:fig8}, these statistics are well maintained, except for the mean sequence duration in the MIMIC-III data where the data is spread very widely over time with a large standard deviation between transition times. Therefore, in situations where no time regularities exist, especially when there are long transition times, it is very difficult to maintain the time dimension accurately.
\begin{figure}[htbp!] 
\centering    
\includegraphics[width=\linewidth ]{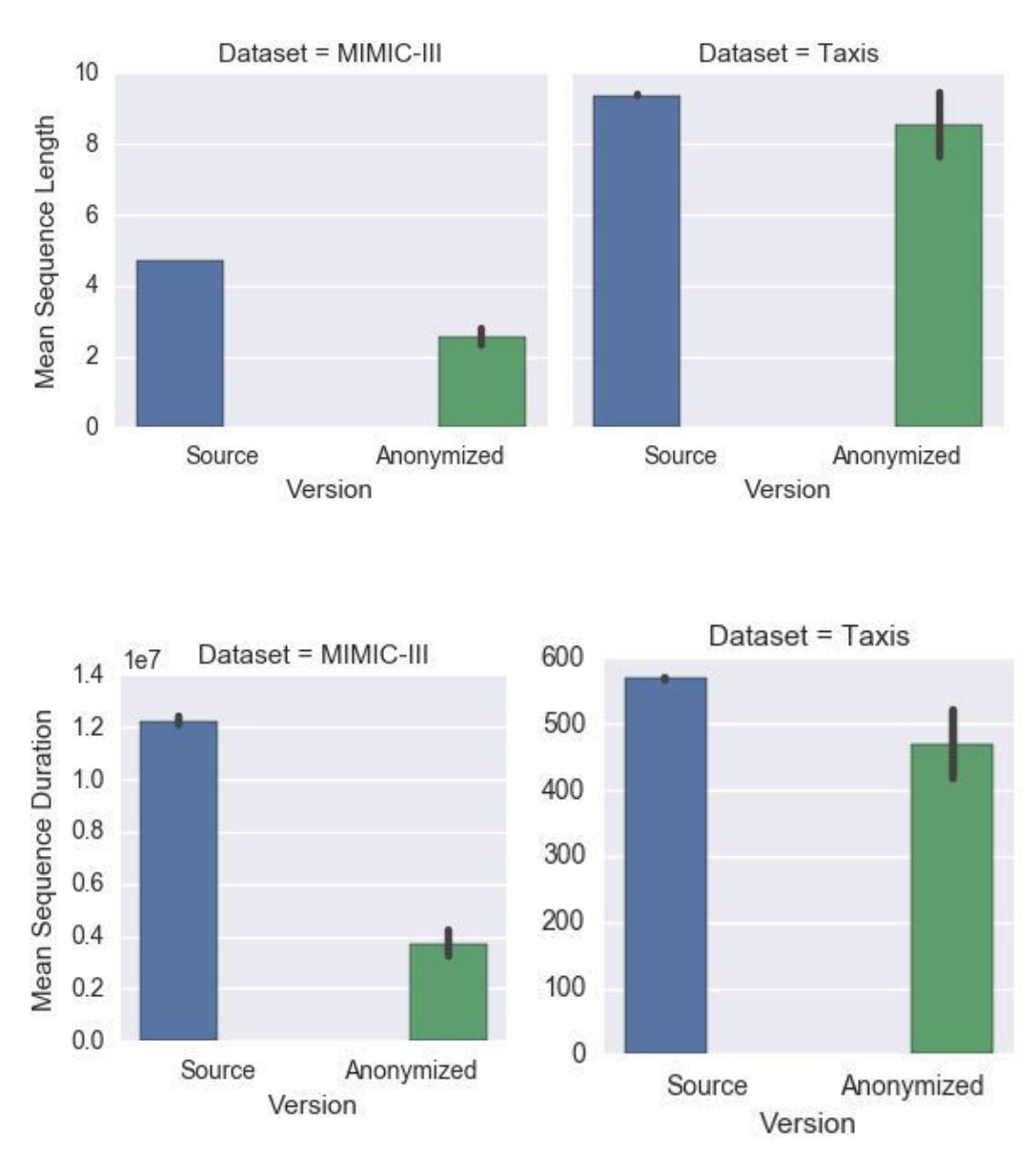}
\caption[Mean sequence length and duration for source and anonymized data versions]{Mean sequence length (above) and mean sequence duration (below) for source (blue) and anonymized (green) data versions.}
\label{fig:fig8}
\end{figure}

We now examine the preservation of state-related statistics. Fig.~\ref{fig:fig9} presents state frequencies. In both datasets, we see some similarity between state frequencies. The Markov model preserves the state frequencies best for both datasets. XGBoost also performed well for the Taxis dataset.
\begin{figure}[htb!] 
\centering    
\includegraphics[width=\linewidth ]{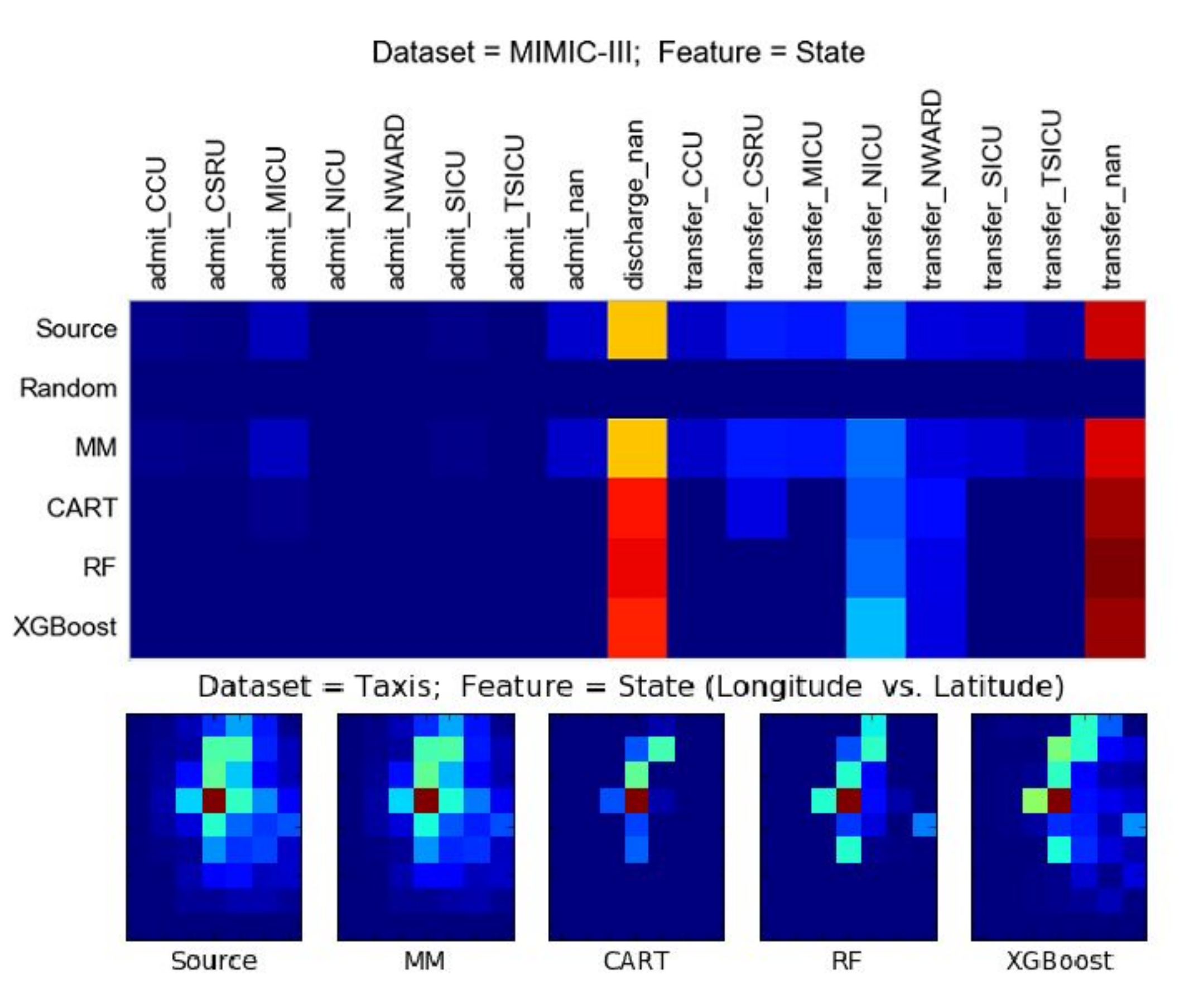}
\caption[Hitmaps of state frequencies]{Heatmaps of state frequencies. For the MIMIC-III dataset (top), each row in the map represents the state space, and the colors indicate the states' frequencies. For the Taxis dataset (bottom), the states are presented according to their positions (x and y coordinates). The left matrix describes the source version of the data, and the other matrices describe the anonymized versions generated using the various algorithms.}
\label{fig:fig9}
\end{figure}
\subsection{Comparing PrivGen with State of the Art Techniques}
\label{section8-2}
In this set of experiments we compare the data utility achieved by our suggested method, with the utility achieved by state of the art techniques for privacy preserving sequence data publishing. Since our method is not intended to preserve data at the record level but rather to preserve higher level data characteristics, with an emphasis on sequencing patterns, we evaluate our method using a fundamental analysis task in sequence data mining, as was used for the evaluation of state of the art methods for sequence data publishing \cite{chen2012differentially_1,zhang2016privtree,li2018privts}.

The examined task is to identify the top-k frequent subsequences in $D$, that is, the $k$ subsequences that appear the largest number of times in $D$. We calculate the precision of the top-k subsequences returned by the published data ${D}'$ as follows, when considering the most frequent subsequences in the original data as ground truth:

\begin{equation}
precision=\frac{topK(D)\cap topK(D')}{k}\end{equation}

We use two real sequence datasets: mooc\footnote{\url{https://www.kddcup2015.com/competition/kddcup2015}} and msnbc \cite{sanchez2016utility}. After merging co-occurring events, mooc contains 79,186 learner behavior sequences on MOOC platforms, and the behaviors are divided into seven categories: task work, video viewing, access to other objects, access to a wiki course, access to the course forum, navigation to another part, and closing the webpage. The msnbc dataset consists of 989,818 URL category sequences, each of which matches the user’s browsing history at msnbc.com over a 24-hour period. Table~\ref{tab:3} shows the main statistics for the mooc and msnbc datasets.

\begin{table*}[!htb]
\centering
\caption{Data description for the second set of experiments.}
\begin{tabular}{lccccl}
\hline\noalign{\smallskip}
Dataset & Sequences & Records & Seq-size 
& Seq-size
& Description\\
 &  &  & (avg) 
& ($80\%$)
& \\
\noalign{\smallskip}
\hline
\noalign{\smallskip}
$mooc$ & 79,186 & 6,234,900 & 7 & 105 & Users' behavior sequences on a\\
 & & & & &MOOC website\\
$msnbc$ & 989,818 & 4,698,794 & 4.75 & 16 & Users' Web navigation histories\\
 & & & & &on a news portal\\
\hline
\label{tab:3}
\end{tabular}
\end{table*}

We compare PrivGen against three differential privacy preserving techniques: PrivTree \cite{zhang2016privtree}, n-gram \cite{Chen2012}, and EM \cite{mcsherry2007mechanism}. These methods require that the maximum sequence length in the input data be bounded by a constant, which we set to be roughly the $95\%$ quantile of the sequence lengths in the input data, as done in previous work \cite{Chen2012,zhang2016privtree}. N-gram is known as the state of the art solution for privacy preserving sequence data publishing, and it is based on a variable length n-gram model. It requires a predefined threshold $n_{max}$ for the maximum length of n-grams; we set $n_{max}$ = 5, as suggested in \cite{chen2012differentially_1}. PrivTree is a histogram construction algorithm that adopts hierarchical decomposition but completely eliminates the dependency on a predefined $n_{max}$. EM is a standard application of the exponential mechanism \cite{mcsherry2007mechanism} in our context, as described in \cite{zhang2016privtree}. 

We set PrivGen to generate smaller synthesized data versions with $10\%$ of the sequences in the original data, and the maximal sequence size set to the $80\%$ quantile of the sequence lengths in the original data, as presented in Table~\ref{tab:3}. For each of the two datasets examined we generated 10 synthesized versions using PrivGen and performed the task of preserving the top-k subsequences for each version; we report the average precision values for this task. 

\begin{figure}[htbp!] 
\centering    
\includegraphics[width=\linewidth ]{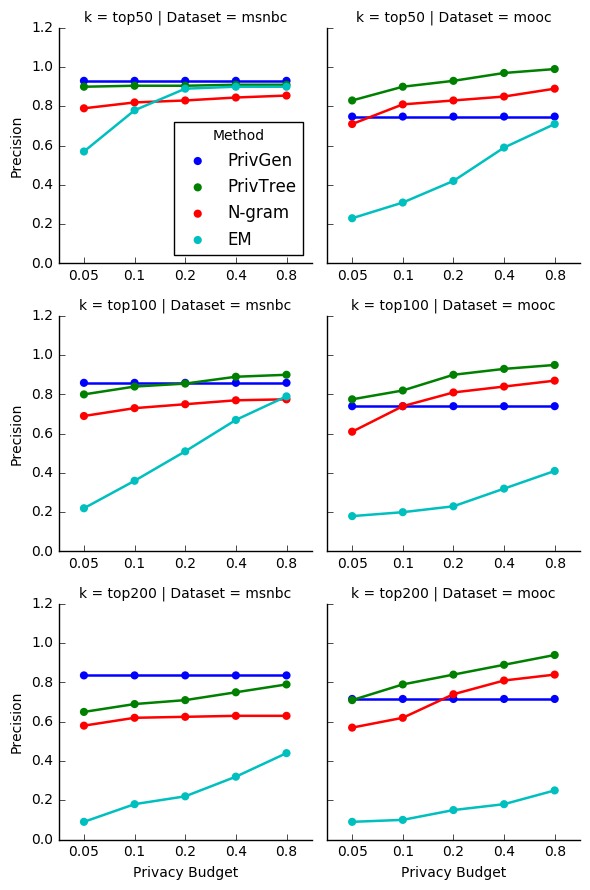}
\caption[Precision of top-k frequent subsequence versus different privacy budgets]{Precision of top-k frequent subsequence mining (y-axis) versus different privacy budgets (x-axis), for the two datasets examined.}
\label{fig:fig10}
\end{figure}
Figure~\ref{fig:fig10} shows the precision for top-k subsequence mining for each method as a function of the privacy budget ($\epsilon$). The precision of PrivGen remains unchanged for all $\epsilon$ values, since is does not enforce differential privacy. In the msnbc dataset, for all three top-k sizes examined, PrivGen has a significantly higher precision than the other techniques; for the top-100 this is only true for $\epsilon$ values lower than 0.2. That means that for generating datasets that guarantee high privacy levels (lower $\epsilon$) PrivGen preserves data utility better in the sense of preserving the most frequent subsequences. In the mooc dataset, PrivGen’s precision is close to that of the other techniques for low $\epsilon$ values. It seems that PrivGen better preserves frequent subsequences in the case of short sequences such as in the msnbc dataset, but when the sequences are longer, the precision is compromised. It may be necessary to use a higher order Markov model instead of the second-order model used in this work, in order to better capture patterns among longer sequences. Yet, as we seek to include more frequent subsequences (top-200), PrivGen shows slightly better results for low $\epsilon$ values, even in this more complex scenario.

Although it cannot be used for sequence publication, we should also mention PrivTS \cite{li2018privts}, a recently proposed solution for mining frequent T-sequences with a $\epsilon$-differential privacy guarantee. As PrivTree was found to outperform PrivTS in preserving the most frequent subsequences, PrivGen could also be considered a potential solution for mining frequent T-sequences with privacy guarantees.

Based on the results presented in Fig.~\ref{fig:fig10}, we conclude that PrivGen is a more preferable solution than $\epsilon$-differential privacy-based solutions for low $\epsilon$ values or for protecting privacy regardless of the $\epsilon$ settings.
\section{Conclusions}
\label{section9}
In this paper we studied the challenges and solutions for anonymizing sequence data, defined as a set of ordered data points. Sequence data is very common in many information systems, representing widespread personal data such as medical records, behavioral traces, location traces, and so on. This makes finding privacy solutions to sequential data extremely important. While sequential data is very common, there is a shortage of privacy solutions for sequential data. Applying privacy solutions to sequential data raises some challenges in light of the multidimensional nature of this data, as instances constructed of the ordered data might be highly unique. 

Synthetic datasets generated via differential privacy based techniques were found to be better than traditional anonymization techniques and are recommend to be used instead of anonymized ones when possible, to avoid the arms race between deidentification and reidentification. However, we point the shortcoming of differential privacy, the current leading anonymization technology. The mainstream implementation of differential privacy, such as adding random noise to the data or suppressing properties, do not work well with detailed sequences. Moreover, determining the privacy threshold, which is the key to achieving a balance between privacy and utility, is yet an unsolved problem. 

In this paper, we introduced PrivGen, a new method for generating sequential data based on patterns learned from the source data. The method is based on a number of components, each modeling a different dimension of the data, and together they enable the preservation of a wide variety of patterns and behaviors identified in the source data. We exploit the inverse relationship between the sequence frequency and its reconstruction probability to demonstrate that the synthetic sequences generated by PrivGen guarantee privacy for individuals. We conducted extensive experiments which demonstrated PrivGen’s ability to generate versions of the data that maintain the sequential model, as well as explicit and hidden attribute statistics of the original data. Our experimental results demonstrate that for strict privacy requirements, which are represented by low $\epsilon$ values, PrivGen provides comparable performance to that of state of the art techniques in terms of preserving the top frequent subsequences and can therefore be used when there is uncertainty regarding the amount of the privacy budget. Moreover, in scenarios where sequences are relatively short (like in the msnbc dataset), PrivGen outperforms state of the art techniques for low $\epsilon$ values and therefore can be used for publishing sequences of higher quality while preserving individuals’ privacy.

In this paper we provide a theoretical analysis of the data synthetization process, assessing the re-identification risk based on the uniqueness of the source data and the reconstruction probability of a source record. An interesting future research that may advance existing privacy solutions would be to examine how to incorporate this concept within the differential privacy solution so that noise is adjusted to the level of re-identification risk. It would also be interesting to explore how deep architecture-based methods could be used for generating multidimensional sequence data.

\bibliographystyle{unsrt}  
\bibliography{template}

\end{document}